\documentclass[11pt]{article}
\usepackage{amssymb}
\usepackage{amsmath}
\usepackage{amsthm}
\usepackage[dvips]{graphics,color}
\usepackage{epsfig}
\usepackage{rotating}
\usepackage{a4wide}

\def\E{\mathbb{E}}
\newcommand{\pr}{\stackrel{\mathcal{P}}{\longrightarrow}}
\def\var{\text{Var}}
\def\cov{\text{Cov}}

\newtheorem{thm}{Theorem}
\newtheorem{alg}{Algorithm}

\begin{document}

\title{Multivariate Control Charts based on Bayesian State Space Models}

\author{K. Triantafyllopoulos\footnote{Department of Probability and Statistics,
Hicks Building, University of Sheffield, Sheffield S3 7RH, UK,
Email: {\tt k.triantafyllopoulos@sheffield.ac.uk}, Tel: +44 114 222
3741, Fax: +44 114 222 3759.}}

\date{31 January 2006}

\maketitle

\begin{abstract}

This paper develops a new multivariate control charting method for
vector autocorrelated and serially correlated processes. The main
idea is to propose a Bayesian multivariate local level model, which
is a generalization of the Shewhart-Deming model for autocorrelated
processes, in order to provide the predictive error distribution of
the process and then to apply a univariate modified EWMA control
chart to the logarithm of the Bayes' factors of the predictive error
density versus the target error density. The resulting chart is
proposed as capable to deal with both the non-normality and the
autocorrelation structure of the log Bayes' factors. The new control
charting scheme is general in application and it has the advantage
to control simultaneously not only the process mean vector and the
dispersion covariance matrix, but also the entire target
distribution of the process. Two examples of London metal exchange
data and of production time series data illustrate the capabilities
of the new control chart.

\textit{Some key words:} time series, SPC, multivariate control
chart, state space model, EWMA.

\end{abstract}

\section{Introduction}\label{s1}

In the last decades multivariate Statistical Process Control (SPC)
has received considerable attention, since in practice many
processes are observed in a vector form (Montgomery$^1$). Univariate
control charts have been extensively discussed in the literature
(Montgomery$^1$, Box and Luce{\~ n}no$^2$, del Castilo$^3$) and many
efforts have been devoted to upgrading the control charts for: (a)
cases of correlated univariate processes; and (b) cases of
multivariate uncorrelated processes.

Multivariate control charting has been discussed in many studies,
e.g. Tracy {\it et al.}$^4$, Liu$^5$, Kourti and MacGregor$^6$,
Mason {\it et al.}$^7$, Vargas$^8$, Ye {\it et al.}$^9$ and
Pan$^{10}$ among many others. Review papers on multivariate
control charts include Lowry and Montgomery$^{11}$, Sullivan and
Woodall$^{12}$, Montgomery and Woodall$^{13}$, Bersimis {\it et
al.}$^{14}$ and Yeh {\it et al.}$^{15}$. Most of the current
research has been focused on the Hotelling's $T^2$ control chart
and the multivariate EWMA control chart for controlling the
process mean. Yeh {\it et al.}$^{16}$, Surtihadi {\it et
al.}$^{17}$, Cheng and Thaga$^{18}$ and Costa and Rahim$^{19}$
propose and study multivariate EWMA and CUSUM control charts to
control the dispersion of a multivariate process. As stated before
univariate control charts for autocorrelated processes have been
discussed in the literature (Montgomery$^1$, Box and Luce{\~
n}no$^2$), however, for multivariate processes the general focus
has been placed to uncorrelated processes. Dyer {\it et
al.}$^{20}$, Jiang$^{21}$, Kalgonda and Kulkarni$^{22}$ and
Noorossana and Vaghefi$^{23}$ consider multivariate control
charting for autocorrelated processes based on
autoregressive-moving-average (ARMA) time series models and the
$T^2$ and multivariate CUSUM control charts are illustrated. Pan
and Jarrett$^{24}$ build a multivariate $T^2$ control chart for
the forecast errors of the process. They consider a state-space
approach for modelling the underlying process and they point out
that the problem of monitoring multivariate processes is a problem
of multivariate time series forecasting as well as a problem of
control charting. Some forms of Bayesian control charts, known
also as adaptive or dynamic control charts, are discussed in
Tagaras$^{25}$, Tagaras and Nikolaidis$^{26}$, de Magalh{\~ a}es
{\it et al.}$^{27}$ and in references therein. Adaptive control
charts offer the flexibility and versatility to dynamically change
the sampling size and the sampling interval of a Shewhart control
chart, but they are disadvantaged in that the complexity is
increased and usually the modeller has to resort to Monte Carlo
simulation.

Our aim in this paper is to construct a multivariate control chart
for autocorrelated processes in such a way that the scheme will be
capable to monitor the process mean vector only, the process
dispersion covariance matrix only, or both the process mean vector
and the process dispersion covariance matrix. We propose a new
control chart based on the theory of sequential Bayes' factors
(West and Harrison$^{28}$). First we fit a local level model to
the multivariate process and then we apply a univariate modified
EWMA control chart to the logarithm of the Bayes' factor to
monitor the dispersion of the predictive distribution of the data
from the target distribution. Our model makes use of a
generalization of the Shewhart-Deming model for multivariate
autocorrelated processes (Deming$^{29}$, del Castilo$^3$,
Triantafyllopoulos {\it et al.}$^{30}$).

Section \ref{s2} gives the necessary time series background. The
proposed control chart is discussed in detail in Section
\ref{s2s2}. In Sections \ref{s5s2} and \ref{s5s1} two examples,
consisting of data from the London metal exchange and from a
production of a plastic mould, illustrate the methodology and give
light to the design and implementation of the new control chart.
Concluding comments are given in Section \ref{s6} and the appendix
details a proof of an argument in Section \ref{s2s2}.

\section{Background}\label{s2}

The conventional control charts are based on the Shewhart-Deming
model, e.g. for a $p\times 1$ process vector $y_t$ this model sets
\begin{equation}\label{eq1}
y_t=\mu +\epsilon_t,\quad \epsilon_t \sim \mathcal{N}_p(0,\Sigma),
\end{equation}
where $\mu$ is the process mean vector and $\Sigma$ is the process
dispersion covariance matrix, known also as the measurement
covariance matrix. Here $\mathcal{N}_p(0,\Sigma)$ indicates the
$p$-dimensional normal distribution with mean vector zero and
covariance matrix $\Sigma$. The measurement drift sequence
$\{\epsilon_t\}$ is assumed uncorrelated and this makes the
generating process $\{y_t\}$ an uncorrelated sequence too. In this
paper we extend the above model by considering equation
(\ref{eq1}), but now $\mu$ is replaced by a time-dependent
$\mu_t$, which follows a multivariate random walk model, known
also as local level model (Durbin and Koopman$^{31}$).

Discount Weighted Regression (DWR), which originated in the
path-breaking work of Brown$^{32}$, is a method for forecasting
autocorrelated time series. Considering univariate time series
Ameen and Harrison$^{33}$ developed further DWR for more complex
time series. The reviews of Ameen$^{34}$ and Goodwin$^{35}$
suggest that DWR can model efficiently time series in a wide range
of situations. Triantafyllopoulos and Pikoulas$^{36}$ developed a
multivariate version of DWR and these authors focused on the
estimation of the measurement covariance matrix. In this paper we
consider the DWR method of Triantafyllopoulos$^{37}$ for
multivariate local level models defined by
\begin{equation}\label{eq2}
y_t=\mu_t +\epsilon_t\quad \textrm{and} \quad
\mu_t=\mu_{t-1}+\omega_t,
\end{equation}
where $\epsilon_t \sim \mathcal{N}_p(0,\Sigma)$ and $\omega_t\sim
\mathcal{N}_p(0,\Omega_t\Sigma)$. The scalar $\Omega_t$ is specified
with the aid of a discount factor $\delta$ and the sequences
$\{\epsilon_t\}$ and $\{\omega_t\}$ are mutually and individually
uncorrelated, e.g. $\E (\epsilon_i\epsilon_j')=\E
(\omega_k\omega_{\ell}')=\E (\epsilon_r\omega_s')=0$, for all $i\neq
j$, $k\neq \ell$ and for all $r,s$. Here $\E (\cdot)$ denotes
expectation and $\epsilon_j'$ denotes the row vector of
$\epsilon_j$. The model definition is complete by specifying a prior
distribution $p(\mu_0|\Sigma)$, which is usually the $p$-dimensional
normal distribution, e.g. $\mu_0|\Sigma\sim
\mathcal{N}_p(m_0,P_0\Sigma)$, for some known prior mean vector
$m_0$ and a positive scalar $P_0>0$. It is further assumed that
$\mu_0$ is uncorrelated of all $\omega_t$. For some positive integer
$N>0$, let $y^t =(y_1,y_2,\ldots,y_t)$ be the information set
comprising data up to and including time $t$, for $t=1,2,\ldots,N$.

With the prior $\mu_0|\Sigma\sim \mathcal{N}_p(m_0,P_0\Sigma)$, the
posterior density of $\mu_t|\Sigma,y^t$ is $\mu_t|\Sigma,y^t\sim
\mathcal{N}_p(m_t,P_t\Sigma)$, where $m_t$ and $P_t$ are updated by
\begin{equation}\label{eq:mt}
m_t=m_{t-1}+\frac{P_{t-1}}{\delta+P_{t-1}}e_t=\frac{\delta
m_{t-1}+P_{t-1}y_t}{\delta+P_{t-1}}\quad \textrm{and}\quad
P_t=\frac{1}{\delta+P_{t-1}},
\end{equation}
with $e_t=y_t-\E (y_t|y^{t-1})=y_t-m_{t-1}$ being the one-step
forecast error vector at time $t-1$. Define the residual error
vector $r_t=\E (\epsilon_t|y^t)=y_t-m_t$. For each time $t$ the
estimator $S_t$ of $\Sigma$ is achieved by least squares estimation
as
\begin{equation}\label{eq:var}
S_t=\frac{1}{t}\sum_{i=1}^t
r_ie_i'=\frac{1}{t}\sum_{i=1}^t\frac{\delta
e_ie_i'}{\delta+P_{i-1}},
\end{equation}
after observing that
$$
r_t=y_t-m_t=y_t-m_{t-1}-\frac{P_{t-1}e_t}{\delta+P_{t-1}}=e_t-\frac{P_{t-1}e_t}{\delta+P_{t-1}}
= \frac{\delta e_t}{\delta + P_{t-1}}.
$$
Details of the derivations of $m_t$, $P_t$ and $S_t$ appear in
Triantafyllopoulos and Pikoulas$^{36}$ and
Triantafyllopoulos$^{37}$.

From the above it follows that the one-step forecast density is
$$
y_{t+1}|\Sigma=S_t,y^t\sim
\mathcal{N}_p\left\{m_t,\frac{(\delta+P_t)S_t}{\delta}\right\}
$$
and the corresponding one-step forecast error density is
\begin{equation}\label{eq5}
e_{t+1}|\Sigma=S_t,y^t\sim
\mathcal{N}_p\left\{0,\frac{(\delta+P_t)S_t}{\delta}\right\},
\end{equation}
where $e_{t+1}=y_{t+1}-\E (y_{t+1}|y^t)=y_{t+1}-m_t$.

The adequacy of the model is evaluated via the mean of squared
standard one-step forecast error vector (MSSE), the mean of
absolute percentage one-step forecast error vector (MAPE) and the
mean of absolute one-step forecast error vector (MAE). These
statistics are discussed in Chatfield$^{38}$ and for data
$y_1,y_2,\ldots,y_N$ they are defined by
\begin{gather*}
MSSE=\frac{1}{N}\sum_{t=1}^N\left[(e_{1t}^*)^2~ (e_{2t}^*)^2 ~\cdots
~ (e_{pt}^*)^2\right]',\quad
e_t^*=\left\{\frac{(\delta+P_{t-1})S_{t-1}}{\delta}\right\}^{-1/2}e_t,\\
MAPE=\frac{1}{N}\sum_{t=1}^N\left[\frac{|e_{1t}|}{y_{1t}}~
\frac{|e_{2t}|}{y_{2t}}~ \cdots ~\frac{|e_{pt}|}{y_{pt}}
\right]',\quad MAE=\frac{1}{N}\sum_{t=1}^N\left[|e_{1t}|~|e_{2t}|~
\cdots ~ |e_{pt}|\right]',
\end{gather*}
where $e_t^*$ is the standard one-step forecast error,
$y_t=[y_{1t}~y_{2t}~\cdots~y_{pt}]'$, $e_t=[e_{1t}~e_{2t}~ \cdots
~e_{pt}]'$ and $\{\delta^{-1}(\delta+P_{t-1})S_{t-1}\}^{-1/2}$
denotes the inverse of the symmetric square root of the matrix
$\delta^{-1}(\delta+P_{t-1})S_{t-1}$ based on the spectral
decomposition of symmetric matrices (Gupta and Nagar$^{39}$; pages
6-7). If the model fit is good the MSSE should be close to the
vector $[1~1~\cdots~1]'$, while MAPE and MAE should be as small as
possible in absolute value. Note that the MAPE, as a percentage
statistic, makes sense only for a positive valued process $y_t$,
for all $t$. If this is not the case, then MAPE can not have a
meaningful interpretation and it should be excluded from the
statistical analysis (Chatfield$^{38}$).

\section{The Bayesian Control Chart}\label{s2s2}

\subsection{The Main Idea}

Bayes' factors have been extensively discussed in the statistics
literature and recently they have been applied sequentially for
time series, see e.g. West and Harrison$^{28}$ (Chapter 11).
Salvador and Gargallo$^{40}$ propose a monitoring scheme, based on
Bayes' factors, for multivariate time series, but this approach is
not suitable for control charting, because it is applied in a
model selection problem. In addition to this, most of the Bayesian
time series monitoring (including the work of Salvador and
Gargallo$^{40}$) relies upon simulated based methods and in
particular Monte Carlo simulation. In this paper we favour
non-iterative techniques, because they are faster, more flexible
and easier to apply.

Once we have the distribution (\ref{eq5}) we can construct a
target distribution for the dispersion of $y_t$ from the target
mean and then compare these two distributions. It is well known
(see e.g. Pan and Jarrett$^{24}$) that the forecast errors $e_i$
and $e_j$ $(i\neq j)$ are approximately uncorrelated and the
approximation is so good as $S_t$ is closer to $\Sigma$. Suppose
now that the target mean of $\{y_t\}$ is denoted by $\mu$ and the
process dispersion covariance matrix is denoted by $V$. This
notation is consistent with the Shewhart-Deming model as in
equation (\ref{eq1}), with $V=\Sigma$ so that $\E (y_t)=\mu$ and
$\textrm{Var}(y_t)=\Sigma$, where $\textrm{Var}(y_t)$ denotes the
covariance matrix of $y_t$. Is is assumed that $\mu$ is a
generally unknown vector, but not stochastic. In our model of
equation (\ref{eq2}) we have $\E (y_t|\mu_t)=\mu_t$ and
$\textrm{Var}(y_t|\mu_t)=\Sigma$, but now $\mu_t$ is stochastic
and it also changes with time according to the random walk model
of (\ref{eq2}). We postulate that, if the process is in control,
the one step forecast mean of $y_t$ will be close to the target
mean vector $\mu$ and the forecast covariance matrix of $y_t$ will
be close to the target dispersion covariance matrix $V$. Thus we
can define the target error distribution by $\varepsilon_t\sim
N_p(0,V)$, where $\varepsilon_t=y_t-\mu$ is the process error,
also known in the process adjustment literature (del Castillo$^3$)
as disturbance drift. Here we assume that $V$ is positive definite
matrix, although the proposed approach can be modified when $V$ is
positive semi-definite. According to the above postulate, if model
(\ref{eq1}) describes well the in-control process, density
(\ref{eq5}) should be close to the above target distribution. In
order to find out ``how close" it is, we form the Bayes' factor at
time $t$:
$$
BF(t)=\frac{f_{e}(t)}{f_{\varepsilon}(t)}=\frac{f_{e}(e_t|\Sigma=S_{t-1},y^{t-1})}
{f_{\varepsilon}(\varepsilon_t)},\quad t=1,2,\ldots,N,
$$
where $f_e(t)$ and $f_\varepsilon (t)$ denote the probability
density functions of $e_t$ and $\varepsilon_t$, respectively.

For consistency in the above equation we need to make the convention
$y^0=\emptyset$ (the null or empty set). Since both densities
$f_{e}(t)$ and $f_{\varepsilon}(t)$ are normal we have
\begin{eqnarray}
BF(t)&=&\sqrt{\frac{\delta^p\det{(V)}}{(\delta+P_{t-1})^p\det{(S_{t-1})}}}\exp\left\{
(y_{t}-\mu)'V^{-1}(y_{t}-\mu)/2  \right. \nonumber \\ && \left. -
\delta (y_{t}-m_{t-1})'S_{t-1}^{-1}(y_{t}-m_{t-1})/(2\delta
+2P_{t-1}) \right\}, \label{eq6}
\end{eqnarray}
where $\det(\cdot)$ denotes the determinant of a square matrix. The
Bayes' factor $BF(t)$ takes values from $0$ to $+\infty$. We will
say that the process $y_{t}$ is in control at time $t$, if $f_e(t)=
f_\varepsilon (t)$, or if $BF(t)= 1$; otherwise the process will be
out of control, at this time point. An out of control signal might
be caused because of a mean shift (e.g. when $\E
(y_t|y^{t-1})=m_{t-1}$ is significantly different than $\mu$) or
because of a dispersion shift (e.g.
$\textrm{Var}(y_{t}|\Sigma=S_{t-1},y^{t-1})=(\delta+P_{t-1})S_{t-1}/\delta$
is significantly different than $V$).

\subsection{The Modified EWMA Control Chart for Correlated Data}

A control chart for the Bayes' factor $BF(t)$ can conclude whether
$BF(t)$ is close to 1 and thus whether the process is in control or
not. Since $BF(t)$ is positive valued, it is more convenient to work
with the logarithm of the Bayes' factor
\begin{eqnarray}
LBF(t)&=&\log BF(t) =p\log \delta /2+\{\log \det{(V)}\}/2 -p\{\log
(\delta+ P_{t-1})\}/2 - \nonumber \\ && \{\log \det{(S_{t-1})}\}/2 +
(y_{t}-\mu)'V^{-1}(y_{t}-\mu)/2 \nonumber \\ &&  - \delta
(y_{t}-m_{t-1})'S_{t-1}^{-1}(y_{t}-m_{t-1})/(2\delta
+2P_{t-1})\label{eq7}
\end{eqnarray}
and so we can construct an appropriate univariate control chart
for $LBF(t)$. In order to propose such a chart we need to deal
with two issues: (a) the values of $LBF(t)$ will be serially
correlated and (b) the distribution of $LBF(t)$ might not be
normal.

Considering (a), in our development it is clear that, from the
definition of the $BF(t)$, either the original data $y_t$ are
i.i.d. or auto-correlated, the resulting data $BF(t)$ (or
$LBF(t)$) will be correlated and hence, if the Shewhart or any
other control chart is to be used successfully, they should be
modified appropriately to accommodate for correlated observations.
Many authors have demonstrated that the Shewhart control charts
need to be modified in order to cater for serially correlated
observations (Vasilopoulos and Stamboulis$^{41}$; Schmid$^{42}$).
Similarly, the EWMA needs also to be modified and the resulting
modified EWMA control chart has been discussed in many articles
including Schmid$^{43}$ and VanBrackle and Reynolds$^{44}$.
According to Harris and Ross$^{43}$ ignoring serial correlation
has a stronger effect in EWMA than in the Shewhart control chart,
but as we will see later the EWMA control chart is preferable to
Shewhart, because it is more robust to the assumption of
normality. One could also consider the modified CUSUM chart for
correlated observations, but we will not further discuss this in
the present paper.

Proceeding with (b) one needs to check the assumption of
normality, before applying a modified EWMA (or Shewhart or CUSUM)
control chart. Borror {\it et al.}$^{46}$ studied the ARL
performance of the EWMA and they suggested that the EWMA with a
smoothing parameter equal to 0.05 is very effective, even in the
presence of non-normality of the observations. This result agrees
with Montogomery$^1$ who states for the EWMA ``It is almost a
perfectly non-parametric (distribution free) procedure''.
Maravelakis {\it et al.}$^{47}$ study the robustness to normality
of the EWMA by tabulating characteristics of the run length
distributions (e.g. ARL) for observations generated by several
gamma distributions. These results conclude that, for relatively
low values of the damping parameter of the EWMA and for shifts in
the mean the EWMA control chart can be used, even in the absence
of normality. Moreover, if the process is in-control following a
symmetrical, but not normal, distribution, then the EWMA can be
applied successfully. To the following we look at the empirical
distribution of $LBF(t)$ when the process is in control and when
it is out of control.

We generate 1000 vectors from a bivariate normal distribution
$\mathcal{N}_2(\mu,V)$ with
$$
\mu=\left[\begin{array}{c} 0 \\ 0\end{array}\right] \quad
\textrm{and} \quad V=\left[\begin{array}{cc} 1 & 2 \\ 2 &
5\end{array}\right]
$$
and we generate 1000 vectors for three out of control scenarios.
In scenario 1 we simulate data from $\mathcal{N}_2(\mu_d,V)$
(deviations from the mean $\mu$); in scenario 2 we simulate data
from $\mathcal{N}_2(\mu,V_d)$ (deviations from the covariance
matrix $V$); in scenario 3 we simulate data from
$\mathcal{N}_2(\mu_d,V_d)$ (deviations from both $\mu$ and $V$),
where
$$
\mu_d=\left[\begin{array}{c} 0.5 \\
0\end{array}\right] \quad \textrm{and} \quad
V_d=\left[\begin{array}{cc} 1 & 2.5 \\ 2.5 & 8\end{array}\right].
$$
Figure \ref{hist1} shows the histograms of the $LBF(t)$ for the
above four scenarios (one in control and three out of control
scenarios). From this figure we observe that, although the
distribution of the $LBF(t)$ for the in-control process (panel (a)
in Figure \ref{hist1}) is not-normal, it is roughly symmetric. The
distributions of the $LBF(t)$ for the out of control processes
appear to be slightly skewed, but the histograms are not conclusive.
The important point is the non-normality of the $LBF(t)$ and the
symmetry of the distribution of the in-control process. This enables
us to make use of the modified EWMA control chart, but we note that
the modified CUSUM control chart can also be used. A more formal
confirmation of the non-normality of the distribution of $LBF(t)$
can be carried out by the using standard tests of normality,
however, here the histograms are deemed sufficient to declare the
non-normality of the distribution of $LBF(t)$.

\begin{figure}[t]
 \epsfig{file=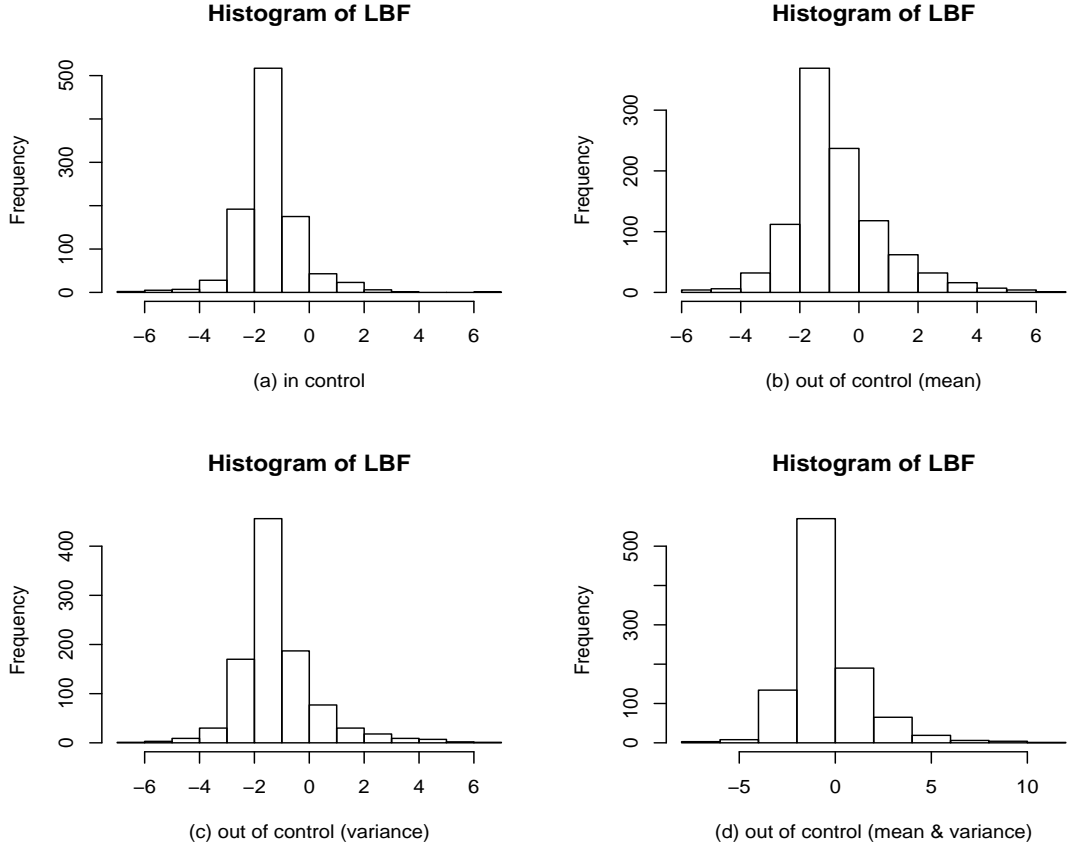, height=12cm, width=15cm}
 \caption{Histograms of the log Bayes' factor $LBF(t)$ for an in-control process (panel (a))
 and out-of-control processes (panels (b)-(d)). The out of control
 scenarios considered are deviations from the mean vector (panel
 (b)), deviations from the covariance matrix (panel (c)) and
 deviations from both the mean vector and the covariance matrix
 (panel (d)).}\label{hist1}
\end{figure}

We use a two phase control scheme; in Phase I the mean $\mu$ and the
covariance matrix $\Sigma$ are estimated and adjustments are applied
if necessary, while in Phase II the EWMA control chart is applied to
detect any changes in the mean of $LBF(t)$. Thus we propose the
algorithm:

\begin{alg}\label{alg3}
There are two phases:
\begin{description}
\item [Phase I:] We fit the DWR model (\ref{eq2}) for a set of
historical data $t=1,2,\ldots,N^*$, with $N^*<N$. We check the
performance and adequacy of the model via the MSSE, MAPE and MAE
over all $t=1,2,\ldots,N^*$ and we possible apply adjustments to the
DWR model, (e.g. adjustments in the mean level) so that we obtain
optimal values $m_{opt}=m_{N^*}$, $S_{opt}=S_{N^*}$,
$\delta=\delta_{opt}$ ensuring that in Phase I the model matches the
in-control process. The modified EWMA control chart is applied so
that control limits are adequately defined according to
pre-specified ARL curves. For this to be designed, a state-space
model for the process $LBF(t)$ needs to be identified and here
simple AR and ARMA modelling will be generally acceptable.
\item [Phase II:] We fit the DWR model with the model components from Phase I
(e.g. $\delta=\delta_{opt}$, $m_t=m_{opt}$, $\Sigma=S_{opt}$ and we
apply a modified EWMA control chart at observations $LBF(t)$ with
the control limits identified at Phase I, for
$t=N^*+1,N^*+2,\ldots,N$.
\end{description}
\end{alg}

In order to apply the modified EWMA control chart we first calculate
the series $z_t$ with observations $x_t=LBF(t)$ as
\begin{equation}\label{ewma}
z_t=\lambda x_t+(1-\lambda)z_{t-1},\quad 0<\lambda\leq 1.
\end{equation}
The parameter $\lambda$ is the EWMA smoothing parameter and as it
is mentioned above, for $\lambda=0.05$ or $\lambda=0.1$ the
control chart is robust to normality. Then, the control limits of
the modified EWMA control chart are
\begin{equation}\label{cl}
\mu_z\pm c\sigma_z,
\end{equation}
where $\mu_z=\E(z_t)$,
$\sigma_z^2=\lim_{t\rightarrow\infty}\var(z_t)$ (asymptotic variance
of $z_t$) and $c>0$ is determined according to the required ARL. For
AR(1) dependence $x_t=\phi x_{t-1}+\nu_t$ and for large $t$, the
asymptotic variance $\sigma_z^2$ is
$$
\sigma_z^2= \frac{\sigma^2\lambda \{1+\phi
(1-\lambda)\}}{(1-\phi^2)(2-\lambda)\{1-\phi (1-\lambda)\}},
$$
where $\nu_t\sim \mathcal{N}(0,\sigma^2)$ and $\sigma^2$, $\phi$
are assumed known. In practice these parameters are estimated at
Phase I. According to Schmid$^{43}$ the asymptotic variance
$\sigma_z^2$ performs better than the exact variance of $z_t$,
which is given in Schmid$^{43}$ and which produces time-dependent
control limits. Most of the literature on this topic focuses on
deriving the variance $\sigma_z^2$ assuming simple time series
models for $x_t$, e.g. as in the above AR(1) or as in the
ARMA(1,1) model considered in VanBrackle and Reynolds$^{44}$.

Algorithm \ref{alg3} can be simplified, if at Phase I, the
quantities $P_t$ and $S_t$ converge to stable values and these
values are determined in Phase I for both phases. This brings up a
well known problem, which has received considerable attention in
the time series literature (see e.g. Durbin and Koopman$^{31}$).
However, for the DWR and similar multivariate models limiting
results for $P_t$ and $S_t$ have not been yet established. The
next theorem (which proof is in the appendix) states that $P_t$
and $S_t$ converge to stable limiting values.

\begin{thm}\label{th2}
In the DWR model (\ref{eq2}) the estimator $S_t$ of the measurement
covariance matrix $\Sigma$ converges in probability to $\Sigma$ and
the non-stochastic scalar parameter $P_t$ converges to the limit
$P=(\sqrt{\delta^2+4}-\delta)/2$, i.e. $S_t\pr \Sigma$ and
$P_t\longrightarrow P$.
\end{thm}

From Theorem \ref{th2} the estimator $S_t$ is consistent and from
the proof of this theorem (given in the appendix), $S_t$ is also
unbiased estimator. Theorem \ref{th2} suggests that $P_{t-1}$ in
the calculation of $LBF(t)$ of equation (\ref{eq7}) can be
replaced by its limit $P$. From equation (\ref{eq:mt}) and Theorem
\ref{th2}, the forecast of $y_t$, $m_{t-1}$ can be approximated by
$$
m_{t-1}=m_0+\frac{P}{\delta+P}\sum_{i=1}^{t-1}e_i =
m_0+\frac{\sqrt{\delta^2+4}-\delta}{\sqrt{\delta^2+4}+\delta}\sum_{i=1}
^{t-1}e_i,
$$
where $P_{t-1}$ of equation (\ref{eq7}) is replaced by $P$. Figure
\ref{fig1} shows how fast $\{P_t\}$ converges to its limit $P$, for
a prior $P_0=1/1000$ and three values of $\delta$. This figure
points out that $P_t$ is bounded above by 1, but for $\delta=0.2$,
this bound is only achieved after $t>13$ (solid line in Figure
\ref{fig1}), while for $\delta=0.9$, this bound is achieved for any
$t>1$ (dotted line in Figure \ref{fig1}). This gives an empirical
indication of the speed of convergence of $\{P_t\}$, for several
values of $\delta$.

\begin{figure}[t]
 \epsfig{file=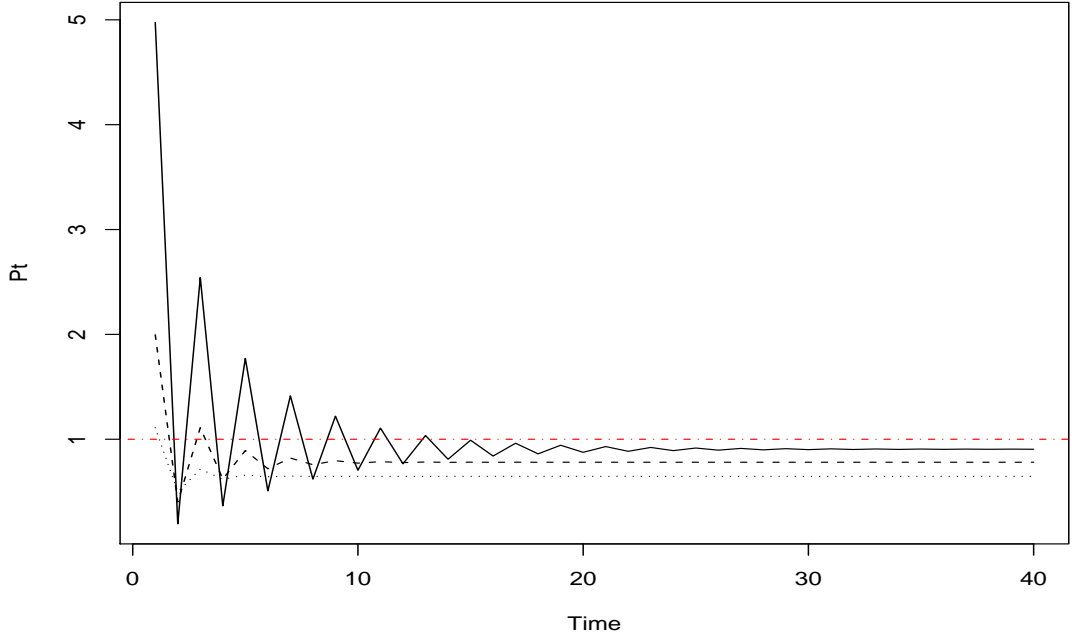, height=10cm, width=15cm}
 \caption{Rate of convergence for the sequence $\{P_t\}$ of Theorem \ref{th2}; the
 solid line plots $\{P_t\}$ for $\delta=0.2$, the dashed line plots $\{P_t\}$ for
 $\delta=0.5$, the dotted line plots $\{P_t\}$ for $\delta=0.9$ and the dashed/dotted line
 is the critical bound of 1.}\label{fig1}
\end{figure}

The limit $P$ is known before the algorithm starts (e.g. $P$
depends only on $\delta$) and, given enough data in Phase I, the
limit $\Sigma$ can be approximated by $\Sigma\approx S_{N^*}$, in
the end of Phase I. This can have an additional benefit on
computational savings, but more importantly it gives a theoretical
justification that the DWR produces a good copy of the process
$\{y_t\}$ and therefore this model is appropriate for the
monitoring part at Phase II of Algorithm \ref{alg3}. For example,
if $P_t$ and $S_t$ were not converging to stable values, no matter
how many data we collected at Phase I, the covariance matrix of
$y_t$ and thus its uncertainty would change over time resulting in
an unstable time series model. False alarms are probable in the
framework of such unstable models, which should be avoided.

In the design and application of the control chart it is important
to suggest values of $m_0$, $P_0$, $\delta$ and $S_0$ and to study
their sensitivity and influence to the performance of the proposed
control chart. Since these suggestions are related to forecasting
as in equation (\ref{eq5}), results on the sensitivity of such
prior parameters follow from Triantafyllopoulos and
Pikoulas$^{36}$ and Triantafyllopoulos$^{37}$. It is worthwhile
noting that, given enough data in Phase I, the values of $m_0$,
$P_0$ and $S_0$ are not critical to the forecast performance, as
in time series modelling prior information is deflated over time.
This is indicated in Theorem \ref{th2} from the fact that $P$ does
not depend on $P_0$. The value of $\delta$ can be critical in
forecasting and a general recommendation is that several values of
$\delta$ (in the range of $(0,1)$) are applied in Phase I and
according to the forecast performance (see Section \ref{s2}) a
value of $\delta$ is decided. One should note that high values of
$\delta$ (e.g. $\delta=0.9$) yield smooth forecasts with low
forecast variances, but these forecasts are sometimes unable to
forecast abrupt changes in the data; low values of $\delta$ (e.g.
$\delta=0.1$) yield more precise forecasts in the presence of
``wild data'', but these forecasts come with increased forecast
variances.

Our proposal for the modified EWMA control chart for the $LBF(t)$
process is motivated from the fact that the observations $LBF(t)$
possess autocorrelation and non-normality. The approach is
model-based, and so a comparison with traditionally used
multivariate control charts, such as the Hotelling's $T^2$ and the
M-EWMA (which are both data-based control charts), is difficult
and in many occasions it can not give justice. Within the
model-based control charting methods, it appears that our approach
can be compared with the residual chart (Pan and Jarrett$^{24}$),
but again the comparisons need to make sure that model uncertainty
(whether for example the DWR is a good model or an alternative
time series model performs better) should be ideally removed
before any comparison is attempted. For example a
miss-specification of a time series model might result to a false
result in the comparison of the competing control charts. From our
experience the DWR works generally well (since it is a
generalization of the Shewhart-Deming model), but this might not
be the case for every multivariate process. We believe that such a
comparison should deserve the length and the detail of a whole
paper and thus here we do not pursue this project. Next we give
two examples illustrating the design and application of the
proposed control chart.

To the above we have assumed that given a process $\{y_t\}$ the
interest is in building a control chart for monitoring
simultaneously the process mean and the dispersion covariance
matrix. However, in some cases the interest is placed on monitoring
the dispersion covariance matrix only. In this case we can modify
the control scheme by considering a modified EWMA control chart of
the log-Bayes' factors of the first order difference process
$z_t=y_t-y_{t-1}$, which from equation (\ref{eq2}) has zero mean.
Control charts based on $\{z_t\}$ will be more robust as compared to
those for $\{y_t\}$, since the uncertainty of monitoring the process
mean of $\{y_t\}$ has been removed.

\section{London Metal Exchange Data}\label{s5s2}

London metal exchange (LME) is the world's premier non-ferrous
metals market trading currently aluminium, copper, lead and zinc,
among other non-ferrous metals. Information on the LME and its
functions can be found in its web site: {\tt
http://www.lme.co.uk}. The review of Watkins and McAleer$^{48}$
explores the recently growing literature on the LME market and
Triantafyllopoulos$^{37}$ discusses the correlation of spot and
future contract prices of aluminium based on the DWR model of
Section \ref{s2}. In this paper we discuss data of spot prices for
the four metals aluminium (variable $\{y_{1t}\}$), copper
(variable $\{y_{2t}\}$), lead (variable $\{y_{3t}\}$) and zinc
(variable $\{y_{4t}\}$).

\begin{figure}[t]
 \epsfig{file=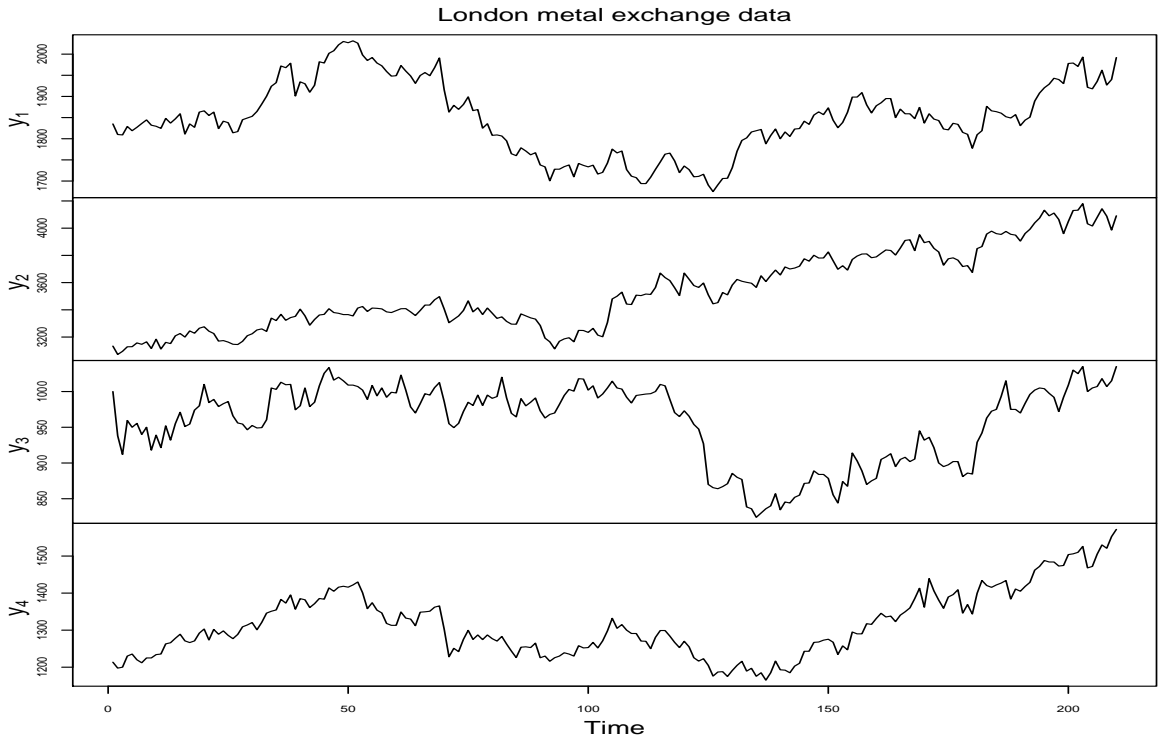, height=10cm, width=17cm}
 \caption{LME data $y_t=[y_{1t}~y_{2t}~y_{3t}~y_{4t}]'$, consisting of aluminium $(\{y_{1t}\})$,
 copper $(\{y_{2t}\})$, lead $(\{y_{3t}\}$ and zinc $(\{y_{4t}\})$ spot prices (in US dollars
 per tonne of each metal).}\label{fig11}
\end{figure}

\begin{figure}[t]
 \epsfig{file=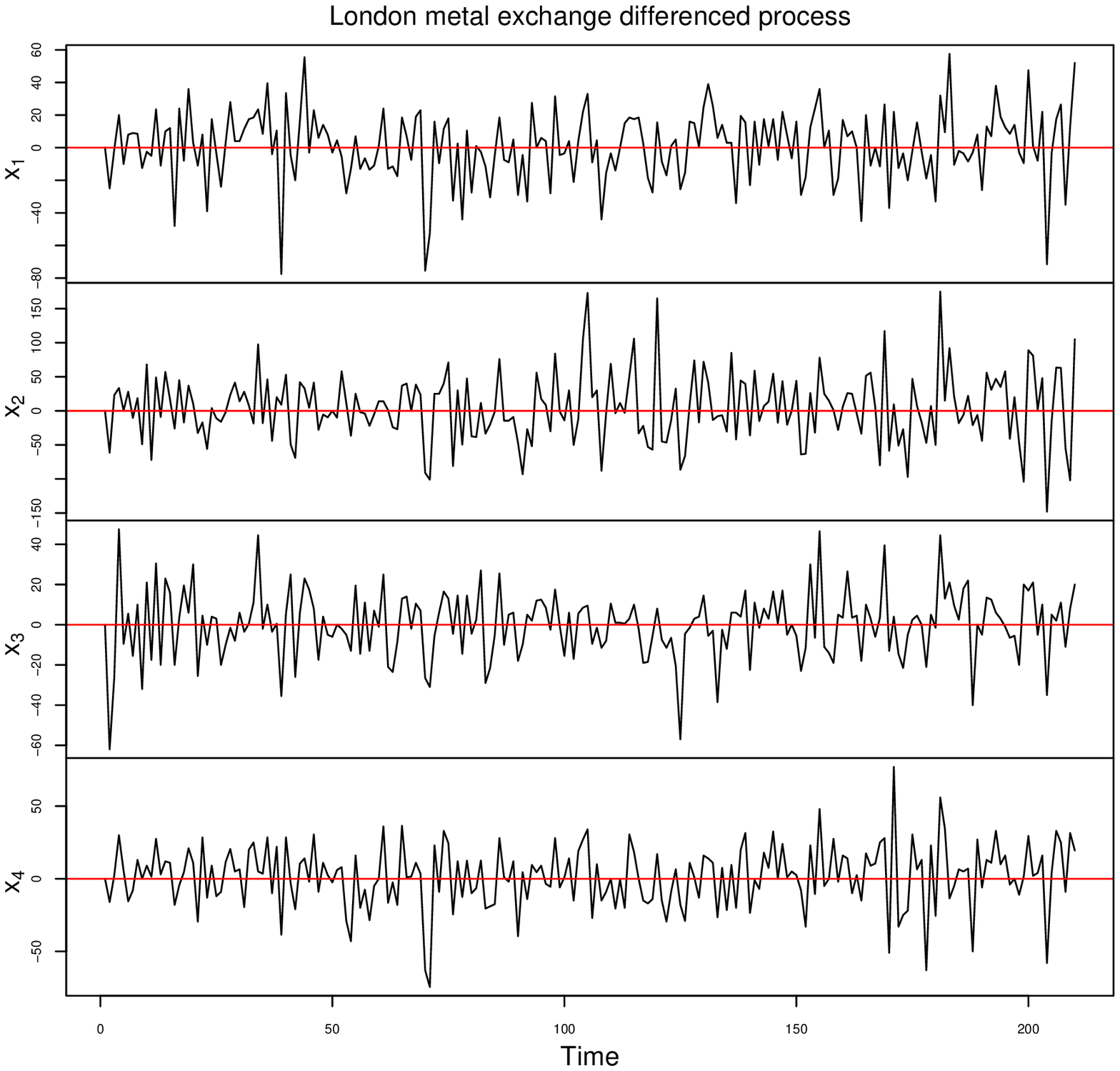, height=10cm, width=17cm}
 \caption{LME differenced process $x_t=[x_{1t}~x_{2t}~x_{3t}~x_{4t}]'$, consisting of
 aluminium $(\{x_{1t}\})$, copper $(\{x_{2t}\})$, lead $(\{x_{3t}\})$ and zinc $(\{x_{4t}\})$.
 The horizontal lines, placed at zero, indicate no volatility.}\label{fig12}
\end{figure}

The data are collected from January 2005 until October 2005 for
every trading day excluding weekends and bank holidays; Figure
\ref{fig11} plots the data. We form the observation vector
$y_t=[y_{1t}~y_{2t}~y_{3t}~y_{4t}]'$ and we are interested in
knowing whether volatility is apparent, for $t=151$ until $t=220$.
In other words we want to know whether from $t$ to $t+1$, the
variability of the observations $y_t$ and $y_{t+1}$ has changed.
This is a major concern to econometricians, because if there is
evidence for volatility, this means there is uncertainty in
investments and ideally the volatility should be understood and
explained. In order to answer this important question we form the
first order difference of the series $\{y_t\}$, defined by
$x_t=y_t-y_{t-1}$, for $t>1$ (Figure \ref{fig12}). Adopting the
usual forecasting strategy of commodity forecasting, given data up
to time $t-1$, the forecast mean of $y_t$ at time $t$ is just the
value of $y_{t-1}$ and so we can write $\E(y_t|y^{t-1})=y_{t-1}$. We
note that the true mean of $x_t$ may not be zero (unless in model
(\ref{eq2}) it is $\mu_t=\mu+\omega_t$), but it is true that
conditionally on $y^{t-1}$ or $y_{t-1}$ we have
$\E(x_t|y^{t-1})=\E(y_t)-y_{t-1}=0$, since
$\E(y_t|y^{t-1})=y_{t-1}$. From Figure \ref{fig12} we observe that
the series $\{x_t\}$ fluctuates around zero and volatility can be
detected as significant deviations from the zero target; such
deviations can be detected with the aid of a control chart of
Section \ref{s2s2}.

\begin{figure}[t]
 \epsfig{file=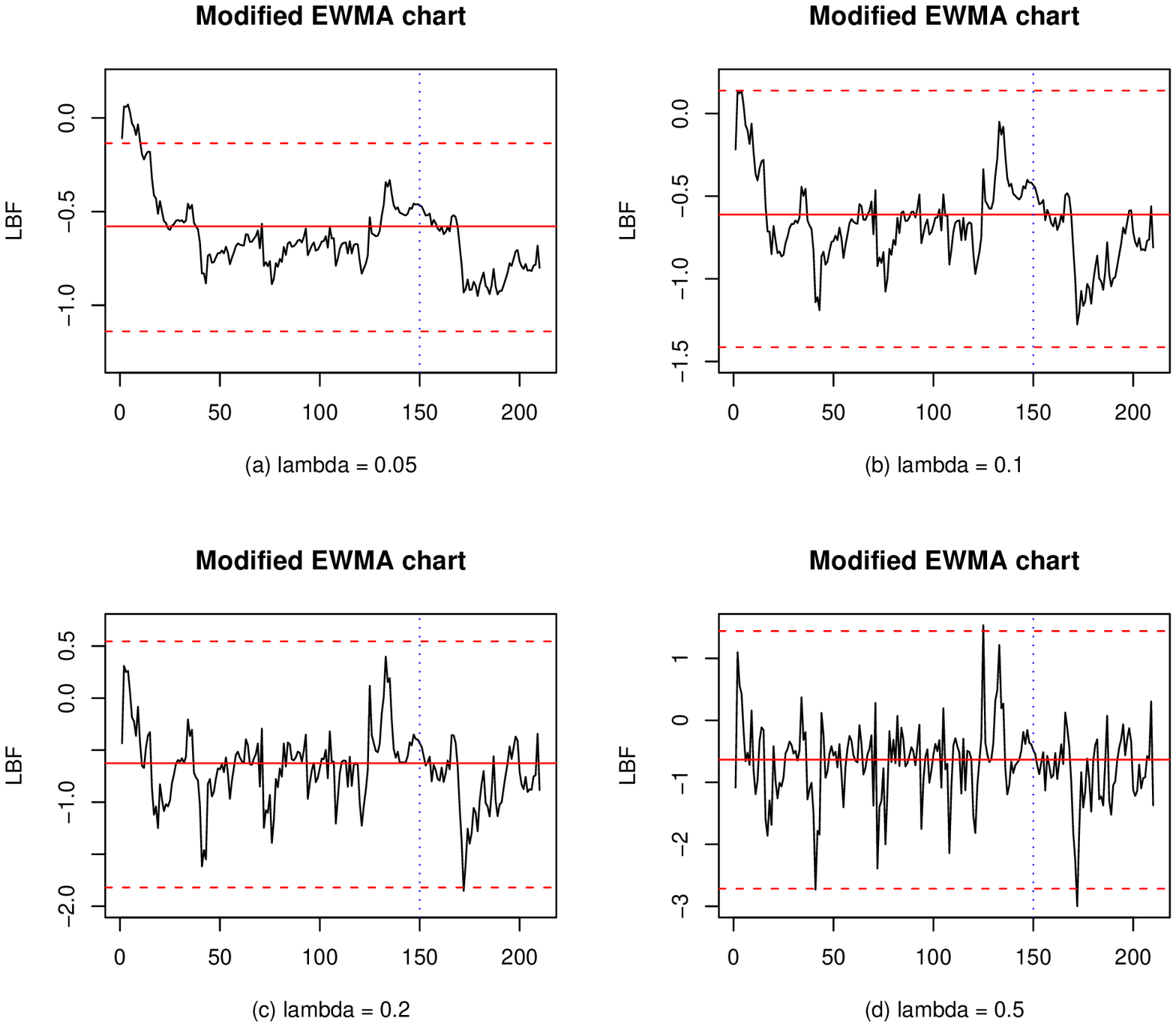, height=9cm, width=15cm}
 \caption{Modified EWMA control chart for the log Bayes' factor of the LME differenced process.
 Plots (a)-(d) show the modified control chart for different values of the
smoothing parameter $\lambda$. In each plot of the panel, the solid
 horizontal line indicates the mean of the EWMA and the dotted
 horizontal lines indicate the control limits; the vertical line separates
Phase I (for $t=1-150$) and Phase II (for $t=151-210$).
}\label{fig13}
\end{figure}

First we need to make sure that the DWR model fits the differenced
series $\{x_t\}$ well. We take $t=1-150$ as Phase I, in which the
adequacy of the DWR model is evaluated. The performance statistics
of Section \ref{s2} are: $MSSE=[0.993~ 1.486~ 0.866~ 1.323]'$ and
$MAE=[18.932~ 45.187~ 14.569~ 19.082]'$, suggesting an acceptable
fit. Of course the MAPE is not available, since $\{x_t\}$ is not a
positive valued process (Section \ref{s2}).

We have designed a modified EWMA control chart for the $LBF(t)$ of
the process $\{x_t\}$ according to the discussion of Section
\ref{s2s2}. Figure \ref{fig13} shows four control charts
corresponding to four values of the EWMA smoothing parameter
$\lambda$. Typically the control chart is robust to normality for
small values of $\lambda$, but for these values the control chart
is only detecting very small drifts in the mean this might not be
desirable. As $\lambda$ increases the modified EWMA control chart
is losing its robustness over normality, but for symmetric process
distributions, such as the empirical distribution of the $LBF(t)$
shown in Figure \ref{hist1}, the EWMA control chart might still be
used for $\lambda=0.5$. The correlation of the $LBF(t)$ is
accounted by the autoregressive model of Section \ref{s2s2} and an
analysis involving the data at Phase I shows that an the
autoregressive parameter $\phi=0.1$ is adequate to capture the
autocorrelation of $LBF(t)$. According to Tables for the $ARL$ of
the modified EWMA control chart (see e.g. Shiau and Hsu$^{49}$) we
choose the value of $c$ in equation (\ref{cl}) so that
$ARL=370.4$, e.g. for $\lambda=0.05$ and $\phi=0.1$ we have
$c=2.469$. The remainder of the control limits are calculated as
in equation (\ref{cl}).

Figure \ref{fig13} shows that the process in Phase II appears to
be in control, for $\lambda=0.05$ and $\lambda=0.1$, while for
$\lambda=0.2$ and $\lambda=0.5$ the control chart returns an out
of control point at $t=172$ (with values $z_{172}=-1.852$ and
$z_{172}=-2.999$, respectively). The mean of the EWMA $z_t$ is
slightly lower than zero, which indicates that, for the entire
process $\{x_t\}$, there will be some deviation of the predictive
density $f_e$ from the target density $f_\varepsilon$. It is up to
the modeller to decide whether such deviation from the target
distribution is worth of declaring the process out of control. In
search of a more automatic approach, one can lift up the whole
control chart so that in Phase I the mean of $z_t$ is exactly
zero. This can be performed automatically, in the end of Phase I,
and this will declare the process in control in Phase II, for
$\lambda=0.05,0.1$, while for $\lambda=0.2,\lambda=0.5$ there is
an out of control point at $t=172$. In Figure \ref{fig13} the
value of $\lambda=0.5$ is rather high to ensuring correct control
limits of the modified EWMA chart (see the relevant discussion in
page \pageref{cl}); here the chart with $\lambda=0.5$ is mainly
shown for comparison purposes with the charts with lower values of
$\lambda$, but in practice we suggest that $\lambda$ does not
exceed 0.2, unless there is strong evidence to support the
assumption of normality for the distribution of $LBF(t)$. It is
worth pointing out that the concentration of consecutive EWMA
values under the mean in Phase II is causing warning, which is
apparent in all charts. The phenomenon is more apparent in the
charts for $\lambda=0.05$ and $\lambda=0.1$ and it can suggest the
out of control state of the process at $t=172$,which is apparent
in the charts with $\lambda=0.2$ and $\lambda=0.5$. The
interpretation of the out of control signal at $t=172$ can not be
done just by looking at Figure \ref{fig12} and more dedicated
methods of out of control variable identification need to be
employed, see e.g. Bersimis {\it et al.}$^{14}$.

\section{Production Time Series Data}\label{s5s1}

In an experiment of production of a plastic mould the quality is
centered on the control of temperature and its variation. For this
purpose five measurements of the temperature of the mould have
been taken, for $276$ time points. The experiment is fully
described in Pan and Jarrett$^{24}$ and these authors show that
this 5-dimensional production process $\{y_t\}$ is both
autocorrelated and serially correlated including both vector
autoregressive and moving average terms. These authors use a
vector state space charting approach based on the Hotelling
control chart resulting on 12 out of control signals at Phase II
(time points from $t=181$ to $t=220$) and hence concluding that
the process falls badly out of control at Phase II.

\begin{figure}[t]
 \epsfig{file=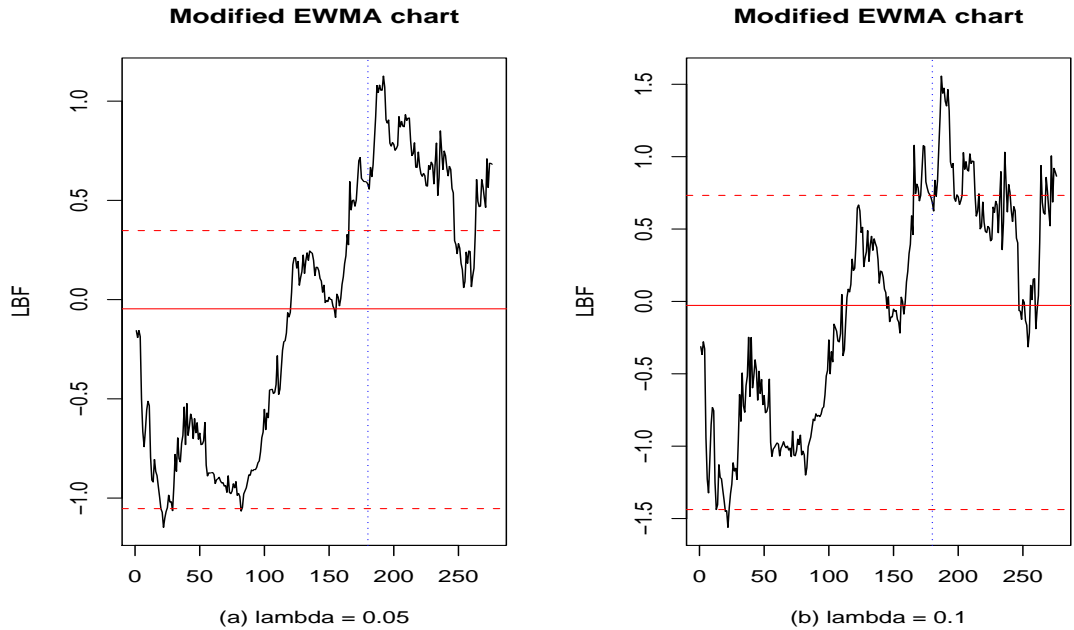, height=9cm, width=15cm}
 \caption{Modified EWMA control chart for the log Bayes' factor of the Production process. Plots (a)-(b)
 show two charts for values of the smoothing parameter $\lambda=0.05$ and $\lambda=0.1$. For
 both plots, the solid horizontal
 line indicates the target mean 0 and the dotted horizontal lines indicate the
 control limits; the solid vertical line separates Phase I (for $t=1-180$) and Phase II
 (for $t=181-276$).}\label{fig10}
\end{figure}

We have used the data at Phase I (time points $t=1-180$) in order to
estimate the target mean vector $\mu=[208.245~ 153.638~ 53.063~
-22.742~ 16.126]'$ (as the average of each $y_{it}$: $t=1-180$) and
the dispersion covariance matrix
$$
V=\left[\begin{array}{ccccc} 0.168 & -0.001 & 0.633 & -0.438 & 0.015
\\ -0.001 & 0.023 & -0.017 & 0.006 &
-0.002 \\  0.633 & -0.017 & 25.621 & -15.658 & 0.453 \\
-0.438 & 0.006 & -15.658 & 14.181 & -0.596 \\ 0.015 & -0.002 & 0.453
& -0.596 & 0.951 \end{array}\right]
$$
(as the sample covariance matrix of each $y_t$: $t=1:180$), where
$y_t=[y_{1t},y_{2t},y_{3t},y_{4t},y_{5t}]'$. The DWR fits well with
$MSSE=[0.855~ 0.950~ 0.992~ 1.161~ 0.996]'$, which is close to
$[1~1~1~1]$. The other two performance statistics are $MAE=[1.378~
0.899~ 4.450~ 3.316~ 0.945]'$ and $MAPE=[0.007~ 0.006~ 0.089~ -~
0.059]'$, where for $\{y_{4t}\}$ the ``--'' indicates that the MAPE
is not available, since this variable is not positive valued (see
the relevant discussion for MAPE in Section \ref{s2}). The above
performance statistics suggest that the model fit is good and
therefore we can proceed with control charting at Phase II
($t=181-279$).

The first thing to do is to find a suitable AR(1) model for the
process $LBF(t)$. A suitable model is the AR(1): $LBF(t)=-4.624 +
0.062 LBF(t-1)+\nu_t$. According to the discussion above, we
remove the intercept $-4.624$ so that we can obtain a in-control
process in Phase I. Thus we design the modified EWMA control chart
for $LBF(t)+4.624$. Again we use tables for the modified EWMA
control chart and for $\lambda=0.05$ the resulting control chart
is given in Figure \ref{fig10}. This figure agrees with the
residual chart of Pan and Jarrett$^{24}$, that finds the process
in Phase II out of control for most of the data points. In Phase I
chart of panel (b) of Figure \ref{fig10} gives one out of control
point, which is in agreement with Pan and Jarrett$^{24}$, but in
panel (a) of Figure \ref{fig10} the control chart detects more out
of control points in Phase I. The EWMA control chart is robust to
non-normality for the low values of $\lambda=0.05$ and
$\lambda=0.1$, but for $\lambda=0.05$ the chart is more sensitive
to small shifts in the mean of $LBF(t)$, resulting to the
detection of out of control points in Phase I. Any out of control
points in Phase I should be immediately investigated and usual SPC
procedures of removing influence of these points in the
calculation of the control limits should be applied
(Montgomery$^1$).

\section{Conclusions}\label{s6}

This paper develops a new multivariate control chart based on
Bayes' factors. This control chart is specifically aimed at
multivariate autocorrelated and serially correlated processes. The
general idea is to form a target distribution, to construct a
predictive density with good forecast ability and then to apply a
univariate control chart for the logarithm of the Bayes' factor of
the predictive error density against the target error density.
Although in this paper, for simplicity, we have considered normal
distributions for the target and the predictive densities, in
general application the proposed control charts can be applied
considering other densities too as long as they are available in
analytic form.

We have restricted our discussion to the modified EWMA control
chart, but other control charts such as the modified CUSUM and
non-parametric control charts can be applied. A major advantage of
our approach as compared to other multivariate control charts is
that once we have obtained the log Bayes' factors we can apply any
appropriate univariate control chart. A difficulty appears to be
that the resulting Bayes' factors process is both autocorrelated and
non-normal, but we believe the design of the proposed chart is a
challenge that can attract and motivate further research in this so
important area of statistical process control.

\section*{Acknowledgements}

I should like to thank the editor Erik M{\o}nness and two anonymous
referees for making several valuable suggestions, which considerably
improved the paper.

\renewcommand{\theequation}{A-\arabic{equation}} 
\setcounter{equation}{0}  

\section*{Appendix}  

\begin{proof}[Proof of Theorem \ref{th2}]
First we prove $S_t\pr \Sigma$. It suffices to prove that $S_t$ is
unbiased estimator and that its covariance matrix converges to zero.
From equations (\ref{eq:var}) and (\ref{eq5}) we obtain
$$
\E(S_t)=\frac{1}{t}\sum_{i=1}^t\frac{\delta
\E(e_ie_i')}{\delta+P_{i-1}}=\frac{1}{t}\sum_{i=1}^t\frac{\delta
(\delta+P_{i-1})\Sigma}{(\delta+P_{i-1})\delta}=\frac{1}{t}(t\Sigma)=\Sigma
$$
and so $S_t$ is unbiased for $\Sigma$. For the convergence, let
$\textrm{vech}(\cdot)$ denote the column stacking operator of a
lower portion of a covariance matrix and let $\parallel
\cdot\parallel$ denote a matrix norm defined in a suitable linear
space. From equation (\ref{eq5}) we have
\begin{equation}\label{eq:app1}
\var\{\textrm{vech}(S_t)\}=\frac{1}{t^2}\sum_{i=1}^t\left(\frac{\delta}{\delta+
P_{i-1}}\right)^2 \var\{\textrm{vech}(e_ie_i')\}.
\end{equation}
From equation (\ref{eq5}) $e_i$ follows a $p$-variate normal
distribution and so by writing $e_i=[e_{i1}~e_{i2}~\cdots~e_{ip}]'$,
we have that $\cov(e_{ij},e_{ik})=\E(e_{ij}e_{ik})$ are bounded,
since these expectations are expressed as moments of the
multivariate normal distribution (Triantafyllopoulos$^{50}$). Hence
$\var\{\textrm{vech}(e_ie_i')\}$ has finite elements and so we can
write $\parallel \var\{\textrm{vech}(e_ie_i')\} \parallel < M$, for
some $M>0$. For any $\epsilon>0$ define $t_0=[\epsilon M]$ (the
integral part of $\epsilon M$). From $P_{i-1}>0$ we have that
$\delta/(\delta+P_{i-1})<1$, for all $i=1,2,\ldots t$. Then
\begin{eqnarray*}
\left\| \var\{\textrm{vech}(S_t)\}\right\| &=& \frac{1}{t^2} \left\|
\sum_{i=1}^t \left(\frac{\delta}{\delta+P_{i-1}}\right)^2
\var\{\textrm{vech}(e_ie_i')\} \right\| \\ &\leq & \frac{M}{t^2}
\left\| \sum_{i=1}^t\left(\frac{\delta}{\delta+
P_{i-1}}\right)^2\right\| \\ &\leq &
\frac{tM}{t^2}=\frac{M}{t}<\epsilon,
\end{eqnarray*}
for any $t>t_0$. This shows that
$\lim_{t\rightarrow\infty}\var\{\textrm{vech}(S_t)\}=0$ and so
$S_t\pr \Sigma$.

Proceeding now with $\{P_t\}$ we show that $\{P_t\}$ is a Cauchy
sequence in the real line and hence
$\lim_{t\rightarrow\infty}P_t=P$ exists. To prove that $\{P_t\}$
is a Cauchy sequence, it suffices to prove that
$\lim_{t\rightarrow\infty}|P_t-P_{t-1}|=0$, where $|\cdot|$
denotes absolute value. First we show that exists positive integer
$t_0$ such that for all $t>t_0$ it is $P_t<1$. The proof of this
is by contradiction. Suppose that for all $t_0$ exists $t>t_0$
such that $P_t\geq 1$. Without loss in generality take $t_0=t^*$
and $P_{t^*}=1$. Then we see that
$P_{t^*+1}=1/(\delta+P_{t^*})=1/(\delta+1)<1$,
$P_{t^*+2}=1/(\delta+P_{t^*+1})=(\delta+1)/(\delta^2+\delta+1)<1$
and likewise $P_{t^*+k}<1$, for all $k\geq 1$. So we can pick
$t_0=t^*+1$ so that we can not find any $t>t_0$ with $P_t\geq 1$,
which contradicts the hypothesis. Thus exists $t_0>0$ so that for
all $t>0$ it is $P_t<1$. This in turn implies that
\begin{equation}\label{eqapp1}
\delta+P_{t-1}>1,\quad \forall~ t>t_0.
\end{equation}
From the definition of $P_t$ of equation (\ref{eq:mt}), we obtain
$$
P_t-P_{t-1}=\frac{1}{\delta+P_{t-1}}-\frac{1}{\delta+P_{t-2}}=-\frac{P_{t-1}-P_{t-2}}{(\delta
+P_{t-2})(\delta+P_{t-2})}=\cdots
=\frac{(-1)^{t-1}(P_1-P_0)}{\prod_{i=1}^{t-1}(\delta+P_{t-i})(\delta+P_{t-i-1})}.
$$
Now pick $t_0$ as in (\ref{eqapp1}) and define $M=\min\{
\delta+P_{t-1},(\delta+P_{t-2})^2,\ldots,(\delta+P_{t_0+1})^2\}$
so that $M>1$. Then
$$
|P_t-P_{t-1}|=\frac{|1-\delta P_0-P_0^2|}{\prod_{i=0}^{t_0}(\delta
+P_i)^2 \prod_{i=1}^{t-t_0-2}(\delta+P_{t-1})(\delta +
P_{t-i-1})}<\frac{|1-\delta P_0-P_0^2|}{\prod_{i=0}^{t_0}(\delta
+P_i)^2M^{t-t_0-1}} \rightarrow 0,
$$
since $\lim_{t\rightarrow\infty}M^{t-t_0-1}=+\infty$. This proves
that $\lim_{t\rightarrow\infty}|P_t-P_{t-1}|=0$ and so $\{P_t\}$
is a Cauchy sequence. Thus $\lim_{t\rightarrow\infty}P_t=P$ exists
and from equation (\ref{eq:mt}) we have $P=1/(\delta+P)$, for
which we derive $P=(\sqrt{\delta^2+4}-\delta)/2$, after rejecting
the negative root $P=(-\sqrt{\delta^2+4}-\delta)/2$.
\end{proof}

\end{document}